\documentclass[twocolumn]{IEEEtran}

\makeatletter
\let\IEEEproof\proof
\let\IEEEendproof\endproof
\let\proof\@undefined
\let\endproof\@undefined
\makeatother

\usepackage{report}

\newcommand{\intl}{\int\limits}
\def\bs{\hspace{-4pt}}

\let\proof\IEEEproof
\let\endproof\IEEEendproof

\usepackage{cite}
\begin{document}

\title{\Large\bfseries On the Finiteness of the Capacity of Continuous
  Channels\thanks{This work was supported by AUB's University Research
    Board (URB) and the Lebanese National Council for Scientific
    Research (LNCSR).}}

\author{
  \authorblockN{Jihad Fahs, Ibrahim Abou-Faycal} \\
  \authorblockA{Dept.\ of Elec.\ and Comp.\ Engineering,
    American University of Beirut \\
    Beirut 1107 2020, Lebanon \\
    {\tt jjf03@aub.edu.lb, iaf@alum.mit.edu}}}

\maketitle


\begin{abstract}
  Evaluating the channel capacity is one of many key problems in
  information theory. In this work we derive rather-mild sufficient
  conditions under which the capacity of continuous channels is {\em
    finite\/} and {\em achievable\/}.
 
  These conditions are derived for generic, memoryless and possibly
  non-linear additive noise channels. The results are based on a novel
  sufficient condition that guarantees the convergence of differential
  entropies under point-wise convergence of probability density
  functions.
 
  Perhaps surprisingly, the finiteness of channel capacity holds for
  the majority of setups, including those where inputs and outputs
  have possibly infinite second-moments.
\end{abstract}

\section{Introduction}

Over continuous-alphabets channels, a common belief is that with
``sufficient'' power, one is capable of transmitting at arbitrarily
large rates. Stated differently, if an input of infinite power is
allowed, the channel capacity is infinite. This belief is perhaps
inspired from the well-known Additive White Gaussian Noise (AWGN) and
linear Gaussian channels for example.

However, recent studies have suggested that for some channels this is
not true: even if an infinite power input is allowed, the achievable
rates are not arbitrarily large:
\begin{itemize}
\item In~\cite{Fahs12}, the authors studied a linear additive-noise
  channel where the noise is heavy-tailed --modeled using alpha-stable
  statistics. They showed that even if the input constraint does allow
  for an infinite-power input, the channel capacity is
  finite. Actually, the authors found the optimal input to be
  surprisingly of finite power.


\item In~\cite{Fahs14-1}, the authors studied an additive
  Cauchy-distributed noise, and the constraints did allow as well for
  infinite-power input signals. The capacity was proven to be finite
  despite the fact that the optimal input was found in this case to
  have infinite power.
\end{itemize}

The natural question that arises is: ``under which conditions does one
have a finite channel capacity?'', the answer to which does clearly
depend on the input constraints, but also on the noise statistics. In
this work, we study the interaction between the input constraints, the
input-output relationship and the noise distribution and derive
conditions on the triplet
under which the channel capacity is finite.

This guarantee of finiteness is of high significance as it is
typically the first step one would undertake in order to quantify the
capacity of a channel at hand. Consider for example an additive Gaussian noise
channel where the output $Y$ is related to the input $X$ as
follows:
\begin{equation}
  Y = X + N,
  \label{exam1}
\end{equation}
and where $X$ is independent of the noise $N$. If no
constraints are imposed on $X$, arbitrarily large transmission rates
are achievable. If a second moment constraint is imposed instead, the
capacity is finite. What if a ``weaker'' constraint is
imposed on $X$. Could the rates be arbitrary large?  For illustrative purposes, 
consider the ``weaker" constraint $\E{\ln^2\left(1+|X|\right)} \leq A$ for some $A > 0$. 
This channel~(\ref{exam1}) is equivalent to the channel:
\[
Y = \sgn(U) \left(e^{|U|} -1\right) + N
\] where $U$ now is average power constrained $\E{U^2} \leq A$.
At a first look, it is not clear whether the capacity of such a
channel is finite or not. Indeed, in some sense the channel is
``exponentially amplifying'' the input and by more than what the cost is constraining
it. An appropriately-chosen Cauchy distributed input $X$ will satisfy the
constraint but will have an infinite second moment. The average of
$Y^2$ will be infinite as well. Is the capacity of this channel
finite? Our result provides an unexpected positive answer to this
question.

Theoretical interests aside, it may seem unusual in a Gaussian setup
to impose the constraint $\E{\ln^2\left(1+|X|\right)}$ or any other
type of input constraints that permits $\E{X^2}$ to be
infinite. However, when the channel model features noise distributions
having an infinite second moment, as in the case of some channels 
subject to multiple access~\cite{ElGhannoudi} or radio-frequency~\cite{marcel}
interference, imposing a second moment constraint becomes less sensible; 
such a constraint masks the characterization of the behaviour of the transmission rates
function of the quality of the channel since the channel
signal-to-noise ratio will constantly evaluate to zero. Furthermore,
we note that the usage of constraints allowing the input to have an
infinite second moment has been previously examined within the context
of robust estimation and detection theory~\cite{hub,Shao,gon}.




More formally, the notion of capacity of a discrete memoryless channel
was defined in the early works of Shannon~\cite{Sha48_1,Sha48_2} to be
``the largest'' rate at which one can communicate over a channel with
an arbitrarily low probability of error. Through a coding theorem,
Shannon proved that the capacity is given by the solution to an
optimization problem, whereby the mutual information between the input
and output of the channel is maximized.  When it comes to continuous
channels the inputs of which are potentially constrained, the results
were extended (see for example~\cite{Sha48_2,Gal68, SMITH71}) and the
channel capacity was also tied to a constrained optimization problem.

Naturally, in both setups it is implicitly assumed that the
optimization problem is ``well-defined'', for otherwise relating the
channel capacity to a solution of a maximization of mutual information
is problematic.  In this work we tackle this assumption and provide a
sufficient condition for such an optimization problem to be both {\em
  well-defined\/} and yielding a {\em finite\/} and {\em achievable
  solution\/} for a wide range of channels.



We consider a generic average-constrained channel model where the
noise is additive and absolutely continuous. We prove in
Section~\ref{sc:genericch} that under very mild conditions on the
noise and the constraint, the channel capacity is indeed finite and
achievable.


We start by deriving sufficient conditions that ensure that mutual
information is finite --and hence well-defined-- and we make use of
the extreme value principle~\cite{Luenb} to ensure that the
maximization problem yields a finite and achievable solution. This
could be achieved by enforcing two characteristics:
\begin{itemize}
\item [1-] The input space of feasible distribution functions is
  compact.
\item[2-] The mutual information between the input and the output of
  the channel is continuous in the input distribution function.
\end{itemize}

We emphasize that these two properties are intimately related to the
channel model and the input constraints if any.

The generic model adopted in this work encompasses multiple channel
models found in the literature: We consider input-output relationships
that are possibly non-linear; A generic average cost function
$\mathcal{C}(\cdot)$ is imposed on the input; The absolutely
continuous additive noise has a finite ``super-logarithmic
moment"\footnote[2]{A ``super-logarithmic moment'' is an expectation of
  the form $\E{f(|X|)}$ for some function $f(|x|) =
  \omega(\ln(|x|))$. \\ We say that $f(x) = \omega \left( g(x)
  \right)$ if and only if $\forall \, \kappa > 0, \exists \, c > 0$
  such that $f(x) \geq \kappa g(x), \forall x \geq c$.} as is the case
for Gaussian, uniform, generalized Gaussian, generalized t, Pareto,
Gamma, alpha-stable distributions, and their mixtures.  We show that
whenever the input cost function has a ``super-logarithmic
growth"\footnote[3]{We say that a function $f(x)$ has a
  ``super-logarithmic growth'' whenever $f(|x|) = \omega(\ln(|x|))$.},
the channel capacity is {\em finite\/} and {\em achievable\/}.

Establishing the continuity of mutual information under any
``super-logarithmic'' input constraint is achieved using a novel
result on the {\em convergence of differential entropies\/}. While
numerous studies have tackled this subject (see for
example~\cite{Piera2009, goda2004}), the conditions presented in
Section~\ref{sc:conventr} are among the weakest that insure this
convergence whenever Probability Density Functions (PDFs) converge
point-wise.

The rest of the paper is organized as follows: In
Section~\ref{sc:conventr}, a preliminary theorem concerning the
convergence of differential entropies is listed and proved. The
primary problem and the main result are presented in
Section~\ref{sc:genericch}, where we describe the channel model and
state the conditions under which our result holds. Technical proofs are
derived in Section~\ref{sc:proofs}. Some extensions are listed in
Section~\ref{sc:Ext} and Section~\ref{sc:Conc} concludes the paper.

\section{Convergence of differential entropies}
\label{sc:conventr}

In this section we establish a sufficient condition for the
convergence of differential entropies whenever there is point-wise
convergence of the corresponding PDFs. More precisely, we
prove a theorem that guarantees this convergence under some
rather-mild sufficient conditions. In layman terms, this theorem
states that whenever the PDFs satisfy a super-logarithmic type of
moment, point-wise convergence will imply convergence of differential
entropies. We emphasize that the new conditions are weaker than some
of those derived by Godavarti et al.~\cite[Thm
1]{goda2004}. Alternative conditions found in~\cite[Thm 4]{goda2004}
are not directly related to those presented hereafter.

\begin{theorem}
  \label{thconv}
  Let the sequence of PDFs on $\Reals$, $\{p_m(y)\}_{m \geq 1}$ and
  $p(y)$ satisfy the following conditions:
  \begin{itemize}
  \item[C1-] The PDFs $\{p_m(y)\}_{m}$ and $p(y)$ are uniformly
    upperbounded:
    \begin{equation}
      \exists \, M \in (0,\infty) \, \,\text{s.t.} \sup_{y \in \Reals, m \geq 1} \,\biggl\{ p_m(y), p(y) \biggr\} \leq M.
      \label{UB}
    \end{equation}
    
  \item[C2-] There exists a non-negative and non-decreasing function
    $l: [0,\infty) \to [0,\infty)$,
    such that $l(y) = \omega \left( \ln (y) \right)$\footnotemark[2]
    and
    \begin{equation}
      \sup_{m} \,\biggl\{ \Ep{p_{m}}{l(|Y|)}, \Ep{p}{l(|Y|)} \biggr\} \leq L,
      \label{eq:lcond}
    \end{equation}
    for some positive (finite) value $L$.
  \end{itemize}

  Under these conditions, $h(p_{m}) \rightarrow h(p)$ whenever the
  PDFs $p_m(y) \rightarrow p(y)$ point-wise.
\end{theorem}

Before we prove the theorem, we highlight the importance of 
condition C2 by providing an example where it is not satisfied,
and the theorem does not hold.

\begin{example}
  Consider the sequence of PDFs $\{p_{m}(x)\}_{m \geq 3}$ defined
  on $\mathbb{R}$ as follows:
  \begin{equation*}
    p_{m}(x) = \left\{ \begin{array}{ll} 
        \displaystyle 1 - \frac{1}{\ln m} \quad & x \in  [0;1] \medskip\\
        \displaystyle  \frac{1}{(\ln m)^2} \frac{1}{x} \quad & x \in (1;m].
      \end{array} \right.
  \end{equation*}
  
  This sequence of PDFs converges point-wise to $p(x)$, the uniform
  distribution on $[0,1]$, and condition C1 is satisfied with a
  uniform upperbound $M=1$. Computing the differential entropies, 
  \begin{align*}
  h(p) \,\,\, &= 0\\
    h(p_m)
    &= -\left(1 - \frac{1}{\ln m} \right) \ln \left(1 - \frac{1}{\ln m} \right)\\
    &    \qquad + \frac{2 \ln (\ln m)}{(\ln m)^2} \int_{1}^{m}\frac{1}{x}\,dx  +  \frac{1}{(\ln m)^2} \int_{1}^{m}\frac{\ln x}{x}\,dx\\
    &= -\left(1 - \frac{1}{\ln m} \right) \ln \left(1 - \frac{1}{\ln m} \right) + \frac{2 \ln (\ln m)}{\ln m} 
    + \frac{1}{2}\\
    &\rightarrow \frac{1}{2} \,\,\text{as} \,\, m \rightarrow{\infty}, 
  \end{align*}
  and hence there is no convergence of differential entropies. 
  This is explained by the fact that condition C2 is not satisfied. 
   Indeed, consider any function $l(x)$ that is
  non-negative, non-decreasing and $l(x) = \omega \left( \ln x
  \right)$. By definition, for any $\kappa > 0$, there exists a $c >
  0$ such that $l(x) \geq \kappa \ln x$ for $x \geq c$. Therefore, for any
  $m \geq c$,
  \begin{align*}
     &\Ep{p_{m}}{l(|X|)}\\
    &= \left(1-\frac{1}{\ln m}\right)\int_{0}^{1} l(x)\,dx + \frac{1}{(\ln m)^2} \int_{1}^{m} \frac{1}{x} \, l(x)\,dx \nonumber\\
    &= \left(1-\frac{1}{\ln m}\right)\int_{0}^{1} l(x)\,dx + \frac{1}{(\ln m)^2} \int_{1}^{c} \frac{1}{x} \, l(x)\,dx \nonumber\\
    &\qquad \qquad \qquad \qquad \qquad \qquad \qquad + \frac{1}{(\ln m)^2} \int_{c}^{m}   \frac{1}{x} \, l(x)\,dx\nonumber\\
    &\geq  \left(1-\frac{1}{\ln m}\right)\int_{0}^{1} l(x)\,dx + \frac{1}{(\ln m)^2} \int_{1}^{c} \frac{1}{x} \, l(x)\,dx\nonumber\\
    &\qquad \qquad \qquad \qquad \qquad \qquad \qquad + \frac{\kappa}{(\ln m)^2} \int_{c}^{m}   \frac{1}{x} \, \ln x \,dx \nonumber\\
    &= \left(1-\frac{1}{\ln m}\right)\int_{0}^{1} l(x)\,dx + \frac{1}{(\ln m)^2} \int_{1}^{c} \frac{1}{x} \, l(x)\,dx \nonumber\\
    &\qquad \qquad \qquad \qquad \qquad \qquad \qquad + \kappa\,\frac{(\ln m)^2 - (\ln c)^2}{2 (\ln m)^2}\nonumber\\
    &\geq \kappa\,\frac{(\ln m)^2 - (\ln c)^2}{2 (\ln m)^2},
  \end{align*}
  which is greater than $\frac{3}{8} \kappa$ whenever $m > c^2$.
  Since the inequality holds for any $\kappa > 0$ and $m$ large enough
  then $ \sup_{m} \,\biggl \{ \Ep{p_{m}}{l(|X|)}\biggr\}$ is unbounded
  which violates condition C2. We proceed next to the proof of Theorem~\ref{thconv}.
\end{example}

\begin{proof} 
  We start by noting that the differential entropies $h(p)$ and
  $\{h(p_m)\}_{m \geq 1}$ exist and are finite by virtue of the fact
  that the PDFs are upperbounded and have a finite logarithmic
  moment~\cite[Proposition 1]{rioul2011}. 

  Assume now that the conditions of the theorem hold and that $p_m$
  converges to $p$ point-wise. If the upperbound~(\ref{UB}) $M$ is
  larger than one, consider the change of variables, $Z = MY$ (for
  which $h(Z) = h(Y) + \ln M$,) or equivalently the PDFs,
  \begin{align*}
    d(y) \eqdef \frac{1}{M} p \left( \frac{y}{M} \right), \qquad 
    d_m(y) \eqdef \frac{1}{M} p_m \left( \frac{y}{M} \right), m \geq 1.
  \end{align*}
  These densities are upperbounded by one and the sequence
  $\{d_m(y)\}$ converges point-wise to $d(y)$. Furthermore, the
  function $l'(y) = l(y/M)$ is non-negative, non-decreasing and $l'(y)
  = \omega \left( \ln (y) \right)$. Additionally,
  \[
  \Ep{d_{m}}{l'(|Y|)} = \Ep{p_{m}}{l'(|MY|)} = \Ep{p_{m}}{l(|Y|)} \leq L.
  \]
  The conditions of the theorem therefore hold for the laws $\{ d_m, d
  \}$ and in what follows we assume without loss of generality that $M
  \leq 1$, and the differential entropies are all non-negative .


  Let $\tilde{y}$ be any positive scalar such that $l(\tilde{y}) > 0$,
  and denote by $q(y) = \frac{1}{\pi} \frac{1}{1+y^2}$ the Cauchy
  density. Then, using the convention ``$0 \ln 0 = 0$'' and the
  fact that $y \ln y \geq - \frac{1}{e}$ for $y > 0$, we can write
 \begin{align}
    &- \intl_{|y| \geq \tilde{y}} p(y) \ln p(y)\,dy\nonumber\\  
    &= \,-\intl_{|y| \geq \tilde{y}} p(y) \ln q(y)\,dy + \intl_{|y| \geq \tilde{y}} 
     q(y)\frac{p(y)}{q(y)}\ln\frac{q(y)}{p(y)}\,dy \nonumber\\
    &\leq \, \ln \pi \intl_{|y| \geq \tilde{y}} p(y)\,dy + \intl_{|y| \geq \tilde{y}} 
    \ln \left[1+ y^2 \right] p(y)\,dy \nonumber\\
    &\qquad \qquad \qquad \qquad \qquad \qquad \qquad \quad +  \frac{1}{e} \intl_{|y| \geq \tilde{y}} q(y)\,dy \nonumber\\
    &\leq \, \frac{\ln\pi}{l(\tilde{y})} \intl_{|y| \geq \tilde{y}} l(|y|)\, p(y)\,dy 
    + \intl_{|y| \geq \tilde{y}} \ln \left[ 1 + y^2 \right] p(y)\,dy\nonumber\\
    &\qquad \qquad \qquad + \frac{1}{e} \frac{1}{\ln \left[1+\tilde{y}^2\right]} \intl_{|y| \geq \tilde{y}}  
    \ln\left[1+ y^2\right] q(y)\,dy, \label{nondec88}
  \end{align}
  where equation~(\ref{nondec88}) is due to the fact that $l(\cdot)$
  is non-decreasing. Hence,
  \begin{align}
    &- \intl_{|y| \geq \tilde{y}} p(y) \ln p(y)\,dy\nonumber\\
    &\leq \, \ln \pi \frac{\Ep{p}{l(|Y|)}}{l(\tilde{y})} + 2 \intl_{|y| \geq \tilde{y}} 
    \ln \left[1+\left|y\right|\right] p(y)\,dy \nonumber\\
    &\qquad \qquad \qquad \qquad \qquad \qquad \,\,\,\, + \frac{1}{e} \frac{\Ep{q}{\ln 
        \left[1+Y^2\right]}}{\ln \left[1+\tilde{y}^2\right]} \label{marcoba}\\
    &\leq \,  \frac{L \ln \pi}{l(\tilde{y})} + 2 \sup_{|y| \geq \tilde{y}} \left\{ 
      \frac{\ln \left[1+|y|\right]}{l(|y|)}\right\}
    \intl_{|y| \geq \tilde{y}} l(y)\,p(y)\,dy \nonumber\\
    &\qquad \qquad \qquad \qquad \qquad \qquad \quad + \frac{1}{e} \frac{\ln 4}{\ln 
      \left[1+\tilde{y}^2\right]}  \label{cauch}\\
    &\leq \,  \frac{L \ln \pi}{l(\tilde{y})} + 2 L \sup_{|y| \geq \tilde{y}} \left\{\frac{\ln 
        \left[1+|y|\right]}{l(|y|)}\right\}
    + \frac{1}{e} \frac{\ln 4}{\ln \left[1+ \tilde{y}^2\right]} \label{lastuni},
  \end{align} 
  where equation~(\ref{marcoba}) is justified since $l(\tilde{y})$ is
  positive and $l(y)$ is non-negative. In order to write
  equation~(\ref{cauch}) we use the identity
  $\text{E}_q\left[\ln\left(1+ y^2\right)\right] = \ln
  4$~\cite[Sec.3.1.3, p.51]{kapur}. The supremum in
  equations~(\ref{cauch}) and~(\ref{lastuni}) is finite --and goes to
  $0$-- for $\tilde{y}$ large-enough because $l(y) = \omega\left(\ln y
  \right)$.

  Since the upperbound~(\ref{lastuni}) also holds for any $p_m(y)$,
  then for every $\delta > 0$, there exists a $\tilde{y}>0$ such that for all $m \geq 1$:
  \begin{equation*}
    \left| \, \intl_{|y| \geq \tilde{y}} \bs p_m(y)\ln p_m(y)\,dy\right| < \delta \,\, \& \,\,
    \left| \, \intl_{|y| \geq \tilde{y}} \bs p(y)\,\ln p(y)\,dy\right| < \delta. 
  \end{equation*}
  It remains to show that 
  \begin{equation*}
    \lim_{m \rightarrow +\infty} - \bs \intl_{|y| < \tilde{y}}\bs p_m(y)\ln p_m(y)\,dy
    = - \bs \intl_{|y| < \tilde{y}}\bs p(y)\,\ln p(y)\,dy,
  \end{equation*}
  which is guaranteed by the Dominated Convergence Theorem (DCT) since
  $\left|p_m(y) \ln p_m(y) \right| \leq \frac{1}{e}$ by virtue of the
  fact that $p_m(y) \leq 1$ for all $m$, which completes the proof.
\end{proof}

\section{Sufficient conditions for finiteness of channel capacity}
\label{sc:genericch}

In what follows we derive sufficient conditions for a memoryless
additive-noise channel to have a finite and achievable capacity. More
specifically, we consider a generic discrete-time real and memoryless
noisy communication channel where the noise is additive and where the
input and output are possibly non-linearly related as follows:
\begin{equation}
  Y = f(X) + N,
  \label{eq:generic_ch1}
\end{equation}
where $Y \in \Reals$ is the channel output and where the input
$X$ is assumed to have an alphabet $\set{X} \subseteq
\Reals$. The channel's input is distorted according to the
deterministic and possibly non-linear function $f(x)$.  Additionally,
the communication channel is subjected to an additive noise --that is
independent of the input-- that is absolutely continuous with PDF
$p_{N}(\cdot)$.

Finally, we assume that the input is subject to an average cost
constraint of the form: $\E{\cost{X}} \leq A$, for some $A \in
(0,\infty)$ and where $\mathcal{C}(\cdot)$ is some cost function:
\begin{equation*}
\mathcal{C}: [0,\infty) \longrightarrow \Reals.
\end{equation*}

Accordingly, we define for $A > 0$
\begin{equation}
  \set{P}_A = \Big\{ \text{Prob. distributions } F:\int \cost{x}\,dF(x) \leq A \Big\},
  \label{eq:Set}
\end{equation}
the set of all distribution functions satisfying the average cost
constraint.


The primary question that we would like to answer is whether one can
reliably transmit an arbitrarily large number of bits per use over this
channel. Said differently, are the achievable rates over this channel
bounded? The answer to this question follows from those of
the following two questions:
\begin{itemize}
\item Is the mutual information between a feasible input and the
  corresponding output always finite?
\item If it is the case, can this mutual information be arbitrarily large?
\end{itemize}

A positive answer to the first question allows by the coding
theorem~\cite{cover} to state that the channel capacity is the
supremum of the mutual information $I(\cdot)$ between the input $X$
and output $Y$ over all input probability distributions $F$ that meet
the constraint $\mathcal{P}_A$:
\begin{equation*}
  C = \sup_{F\in \mathcal{P}_A}\,I(F).
\end{equation*} 

For the channel at hand~(\ref{eq:generic_ch1}), we note that the
channel transition probability law is absolutely continuous with
density function given by
\begin{equation}
  p_{Y|X}(y|x) = p_N (y-f(x)), \quad y \in \Reals, \,\,x \in \set{X}.
  \label{eq:py_x}
\end{equation}
and the mutual information may be expressed as~\cite{SMITH71}
\begin{equation}
  \label{mutdefn}
  I(F) \bs \, \eqdef  \bs \iint \bs p_N\left(y-f(x)\right) \ln \left[ \frac{p_N\left(y-f(x)\right)}{p(y;F)} \right] dy\,dF(x),
\end{equation}
where $p(y;F) = \int p_N\left(y-f(x)\right)\, dF(x)$ is the output PDF.

\subsection*{Sufficient conditions}

We make the following rather-mild assumptions:
\begin{itemize}
\item[$\bullet$] \underline{The cost function $\mathcal{C}(\cdot)$:}
  \begin{itemize}
  \item[A1-] The cost function is lower semi-continuous.
  \item[A2-] The cost function is non-decreasing.
  \item[A3-] $\cost{x} = \omega\left(\ln \left|f(x)\right| \right)$.
  \end{itemize}

  Without loss of generality, one may assume that $\mathcal{C}(\cdot)$
  is non-negative. For if it were not, define $\mathcal{C'}(|x|) =
  \mathcal{C}(|x|) - \mathcal{C}(0)$ and adjust the input constraint
  accordingly.
   
  \medskip
\item[$\bullet$] \underline{The function $f(\cdot)$:}
  \begin{itemize}
  \item[A4-] The function is continuous.
  \item[A5-] The absolute value of the function $\left| f(\cdot)
    \right|$ is an non-decreasing function of $|x|$ and $\left| f(x)
    \right| \rightarrow +\infty$ as $|x| \rightarrow +\infty$.
  \end{itemize}
  
  \medskip
\item[$\bullet$] \underline{The noise PDF $p_{N}(\cdot)$:}
  \begin{itemize}
  \item[A6-] The PDF is continuous on $\Reals$.
  \item[A7-] The PDF is upperbounded.
  \item[A8-] There exits a non-decreasing function
    \begin{equation*}
      \mathcal{C}_N: [0,\infty) \longrightarrow \Reals,
    \end{equation*}
    such that $\Ncost{n} = \omega \left( \ln |n| \right)$, and
    \begin{equation*}
      \Ep{N}{\Ncost{N}} = L_N < \infty.
    \end{equation*}

    As an example, this condition holds true for any noise PDF whose
    tail is ``faster'' than $\frac{1}{x \left( \ln x \right)^3}$.
  \end{itemize}

  Conditions A7 and A8 guarantee that the noise differential entropy
  $h_N$, exists and is finite~\cite[Proposition 1]{rioul2011}.

  Since from an information theoretic perspective, the general channel
  model~(\ref{eq:generic_ch1}) is invariant with respect to output
  scaling, we consider without loss of generality that the noise PDF
  is upperbounded by one.

  Also without loss of generality, one may assume that
  $\mathcal{C}_N(\cdot)$ is non-negative. Otherwise, one may adopt
  $\mathcal{C'}_N(|x|) = \mathcal{C}_N(|x|) - \mathcal{C}_N(0)$.
\end{itemize}

The above assumptions are sufficient conditions on the triplet
$f(\cdot)$, $\mathcal{C}(\cdot)$, and $p_N(\cdot)$ that guarantee the
finiteness and the achievability of the capacity of
channel~(\ref{eq:generic_ch1}):
\begin{theorem}
  Under conditions A1 through A8, the capacity of the average-cost
  constrained channel~(\ref{eq:generic_ch1}) is finite and achievable.

  Furthermore, the maximum is achieved by a unique $F^*$ in
  $\mathcal{P}_{A}$ if and only if the output PDF is injective in $F$.
  \label{th:main}
\end{theorem}
We point out that assumptions A4 through A8 are related 
to the channel model at hand and are not ``conditions" per say.
These assumptions are satisfied by the vast majority of common 
models found in the literature.

When thinking in terms of conditions on the input -- controlled by the user,
A1, A2 and A3 are to be considered. Note that these conditions are also common 
to all cost functions found in the literature. While A1 and A2 are rather technical,
the relevance of A3 may be seen in the following example.

\begin{example}
  Consider the linear additive channel~(\ref{exam1}), where now the
  noise $N$ is a uniformly distributed random variable on the interval
  $[0,1)$.

  Let $X_1$ and $X_2$ be two discrete random variables taking integer
  values $k \geq 2$, with respective probability mass functions:
  \begin{equation*}
    p_{X_{1}}(k) = B_1 \, \frac{1}{k (\ln k)^{2}}, \quad p_{X_{2}}(k) = B_2 \, \frac{1}{k (\ln k)^{3}}, 
    \qquad k \geq 2,
  \end{equation*}
  where $B_1$ \& $B_2$ are the normalizing finite constants,
  \begin{equation*}
    B_1 = \left[ \sum_{k =2}^{\infty} \frac{1}{k (\ln k)^{2}} \right]^{-1} \quad 
    B_2 = \left[ \sum_{k =2}^{\infty} \frac{1}{k (\ln k)^{3}} \right]^{-1}.
  \end{equation*} 
  
  Let $Y_1$ and $Y_2$ be the outputs of channel~(\ref{exam1}) whenever
  its inputs are $X_1$ and $X_2$ respectively. Given the placement of
  the mass points, $X_1$ may be perfectly inferred from $Y_1$ and
  $H(X_1|Y_1) = 0$. Similarly $H(X_2|Y_2) = 0$ and therefore the mutual informations
  \begin{align*}
    I(X_1;Y_1) & = H(X_1) - H(X_1|Y_1) = H(X_1) \\
    I(X_2;Y_2) & = H(X_2).
  \end{align*} 
 
  Computing $H(X_1)$ and $H(X_2)$, we obtain:
  \begin{align*}
    H(X_i) & = -\sum_{k \geq 2} p_{X_{i}}(k) \ln p_{X_{i}}(k) \\
    	& =  -\ln B_i + B_i \sum_{k \geq 2} \frac{\ln k + (1+i) \ln (\ln k)}{k (\ln k)^{1+i} } \quad i=1, 2,
  \end{align*}
  which diverges for $i =1$ and converges for $i = 2$. Accordingly, the
  mutual information of channel~(\ref{exam1}) is infinite when the
  input is $X_1$ whereas it is finite for input $X_2$. Note that 
  $\E{\ln X_1}$ is infinite while $\E{\ln X_2}$ is finite, and this example 
  showcases the importance of condition A3 when it comes
  to the finiteness of mutual information. Whenever A3 is not enforced,
  the channel capacity might be infinite as $X_1$ yields an infinite mutual
  information. The theorem states that when the condition is enforced,
  the capacity will be finite. 
\end{example}
 An interesting observation is that both $\E{X_1^2}$ and $\E{X_2^2}$ are infinite, 
 however as inputs to the channel they yield respectively an infinite and a finite mutual
  information. We proceed next to prove Theorem~\ref{th:main}.

\begin{proof}
  The first statement of the theorem is established using the extreme
  value principle which we state for completeness and can be found
  in~\cite{Luenb}:
  \begin{theorem*}
    If $I(\cdot)$ is a {\em real-valued\/}, weak {\em continuous\/}
    functional on a {\em compact\/} set $\Omega \subseteq \set{F}$,
    then $I(\cdot)$ achieves its maximum on $\Omega$.
  \end{theorem*}

  In order to apply this principle, we show in Section~\ref{sc:proofs}
  that the set $\mathcal{P}_{A}$ is {\em compact\/}
  (Theorem~\ref{th:setprop}) and that the mutual information $I(F)$ is
  {\em finite\/} and {\em continuous\/} (Theorems~\ref{th:mutinf1}
  and~\ref{th:mutconc}). Therefore, the capacity of the average-cost
  constrained channel is finite and achievable.

  When it comes to uniqueness, since $\mathcal{P}_{A}$ is {\em
    convex\/} (Theorem~\ref{th:setprop}) whenever $I(\cdot)$ is {\em
    strictly concave\/}, then the maximum
  \begin{equation*}
    C = \max_{F \in \mathcal{P}_{A}}{I(F)},
  \end{equation*}
  is achieved by a unique $F^*$ in $\mathcal{P}_{A}$.

  Knowing that $I(\cdot)$ is {\em concave\/}
  (Theorem~\ref{th:mutconc}), its strict concavity is equivalent to
  the strict concavity of the output differential entropy in
  $p_Y(\cdot)$. This is indeed the case if and only if $p_Y(\cdot)$ is
  injective in $F$.
\end{proof}


The next section is dedicated to the proofs of
Theorems~\ref{th:setprop},~\ref{th:mutinf1} and~\ref{th:mutconc}.

\section{Proofs of the Theorems}
\label{sc:proofs}

We use techniques that have been first developed in~\cite{SMITH71} and
later adopted in various works on mutual information maximization as
in~\cite{fahsj}: Denote by $\set{F}$ the space of all probability
distribution functions on $\Reals$. We adopt weak
convergence~\cite[III-1, Def.2, p.311]{ShiryBook} on $\set{F}$, and
use the Levy metric to metrize this weak convergence~\cite[Th.3.3,
p.25]{hub}. The optimization is carried out in this metric topology.

\subsection*{Optimization set properties}

\begin{theorem}
  Whenever conditions A1, A2, A3 and A5 are satisfied, the set
  $\set{P}_A$ defined in~(\ref{eq:Set})
  is convex and compact.
  \label{th:setprop}
\end{theorem}
  
\begin{proof}

We note first that the theorem was shown to hold for cost functions of the form $\cost{x} = |x|^r$, for $r > 1$ in~\cite{fahsj,Fahs12}. We adopt the same methodologies to generalize the results presented hereafter. 

  \medskip
  \subsubsection*{Convexity} 
  
  Let $F_1$ and $F_2$ be two probability distribution functions in
  $\set{P}_A$, and $\lambda$ some scalar between $0$ and $1$.  Define
  $F = \lambda F_1 + (1-\lambda)F_2$. It is clear that $F$ is a
  probability distribution function because it is non-decreasing,
  right continuous, $F(-\infty)=0$ and $F(+\infty)=1$. Additionally,
  \begin{align*}
    \int_{\Reals} \cost{x} \, dF & = \int_{\Reals} \cost{x} \, d\left(\lambda F_1 + (1-\lambda)F_2\right)\\
    & = \lambda \int_{\Reals} \cost{x} \, dF_1 + (1-\lambda) \int_{\Reals} \cost{x} \, dF_2\\
    & \leq \lambda A + (1-\lambda) A = A.
  \end{align*}
  
  Therefore, $F \in \mathcal{P}_A$ and $\set{P}_A$ is convex.

  \medskip
  \subsubsection*{Compactness}
  
  Consider a random variable $X$ with probability distribution
  function $F \in \mathcal{P}_A$.  Applying Markov's inequality to
  random variable $\cost{X}$ yields,
  \begin{equation*}
    \Pr\{\cost{X} \geq \alpha\} \leq \frac{\E{\cost{X}}}{\alpha}, \quad \forall \alpha>0.
  \end{equation*}
  
  Now let 
  \begin{equation*}
  K = \inf\,\bigl\{ x \in [0,\infty) \,\,\text{s.t.} \,\,\mathcal{C}(x) \geq \alpha \bigr\} + 1,
  \end{equation*}
  which is always greater or equal to 1.  For any finite value of
  $\alpha$, such a $K$ exists since $\mathcal{C}(x)$ increases to
  $+\infty$ as $x \rightarrow +\infty$ by virtue of properties A3 and
  A5. Additionally, since $\mathcal{C}(\cdot)$ is non-decreasing by
  property A2,
  \begin{align*}
    \Pr\{\cost{X} \geq \alpha\} \geq & \, \Pr\left\{|X| > K - 1\right\}
    \geq  \Pr\{|X| \geq K\} \\
    \geq & \, F(-K) + \left[ 1-F(K) \right].
  \end{align*}
  Hence, for all $F \in \mathcal{P}_A$, we obtain
  \begin{equation*}
    F(-K) + \left[ 1-F(K) \right] \leq \frac{\E{\cost{X}}}{\alpha} \leq 
    \frac{A}{\alpha}.
  \end{equation*}
  
  Therefore, for every $\epsilon >0$, there exists a $K_{\epsilon} >
  0$, namely 
  \begin{equation*}
  K_{\epsilon} = \min\,\bigl\{ x \in [0,\infty)
    \,\,\text{s.t.}  \,\,\mathcal{C}(x) \geq (A/\epsilon) \bigr\} +1,
    \end{equation*}
   such that
  \begin{equation*}
    \sup_{F\in\mathcal{P}_A} \left[ F(-K_{\epsilon})+[1-F(K_{\epsilon})] \right] \leq \epsilon.
  \end{equation*}

  This implies that $\mathcal{P}_A$ is {\em tight}~\cite[III-2, Def.2,
  p.318]{ShiryBook}. By Phrokhorov's Theorem~\cite[III-2, Th.1,
  p.318]{ShiryBook}, $\mathcal{P}_A$ is therefore relatively
  sequentially compact and every sequence $\{F_n\}$ of distribution
  functions in $\mathcal{P}_A$ has a convergent sub-sequence
  $\{F_{n_{j}}\}$ where the limit $F^*$ does not necessarily belong to
  $\mathcal{P}_A$. If we prove that $F^* \in \mathcal{P}_A$, the
  latter will be sequentially compact and hence compact since the
  space is metrizable~\cite[Th.28.2, p.179]{Mun}. In order to show
  that the limiting distribution function $F^*$ is in $\mathcal{P}_A$,
  it must satisfy the cost constraint which is the case. In fact,
  \begin{equation*}
    \int \cost{u} \,dF^*(u) \leq  \lim_{n_{j}\to \infty}\inf \, \int \cost{u} \,dF_{n_{j}} \leq A,
  \end{equation*}
  where the first inequality holds because $\mathcal{C}(|u|)$ is
  lower semi-continuous (property A1), and is bounded from below by
  $\mathcal{C}\left(0\right)$ for all $u \in \Reals$ (property
  A2)~\cite[Th. 4.4.4]{klchung}. In addition, the second inequality is
  valid since the sub-sequence $\{F_{n_{j}}\}$ is in $\mathcal{P}_A$
  and therefore satisfies the cost constraint $\forall \,n_{j}$.
  Finally, $F^* \in \mathcal{P}_A$ and $\mathcal{P}_A$ is compact.
\end{proof}

\subsection*{Properties of the mutual information, $I(\cdot)$}

We prove in what follows the finiteness, concavity and continuity of
$I(\cdot)$ on $\mathcal{P}_A$ through Theorems~\ref{th:mutinf1}
and~\ref{th:mutconc}.

\begin{theorem}
  Whenever conditions A3, A7 and A8 hold, the mutual information
  $I(F)$ between the input and output of
  channel~(\ref{eq:generic_ch1}) is finite for all input distribution
  functions $F$ such that $\E{ \cost{X} }$ is {\em finite\/}.
  \label{th:mutinf1}
\end{theorem}

\begin{proof}
  Since $Y = f(X) + N$,
  \begin{equation*}
    \ln \left[1+|Y|\right] \leq \ln \left[1+|f(X)|\right] + \ln \left[1+|N|\right],
  \end{equation*}
  and $ \E{ \ln \left[1+|Y| \right]}$ is finite because both $\E{ \ln
    \left[1+|f(X)|\right] }$ and $ \E{ \ln \left[1+|N|\right]}$ are
  finite (by properties A3 and A8).

  Moreover, and since $p_Y(y)$ is upperbounded (by A7) by one $h_Y(F) =
  -\int\,p(y;F)\,\ln\,p(y;F)\,dy$, the differential entropy of $Y$, is
  well defined~\cite[Proposition 1]{rioul2011} and $0 \leq h_Y(y) <
  +\infty$.

  The differential entropy $h_N$ of the noise being finite (due to
  properties A7 and A8), the mutual information~(\ref{mutdefn}) can
  therefore be written as the difference of two terms:
  \begin{equation}
    I(F)\,=h_Y(F)\,-\,h_{Y|X}(F)\, = h_Y(F) \, - h_N,
    \label{mutdef}
  \end{equation}
  both of which are finite and this completes the proof.
\end{proof}


\begin{theorem}
  Assume that conditions A2 through A8 hold. Under a cost constraint
  \begin{equation*}
  \int \cost{X}\,dF(x) \leq A \qquad A>0,
  \end{equation*} 
 the mutual information $I(F)$ between the input and the output of
  channel~(\ref{eq:generic_ch1}) is concave and continuous whenever
  $\cost{x} = \omega\left(\ln \left|f(x)\right|\right)$.
  \label{th:mutconc}
\end{theorem}

Before we proceed with the proof, we note that under the conditions of
the Theorem, the mutual information $I(F)$ between the input and the
output of channel~(\ref{eq:generic_ch1}) is finite by virtue of
Theorem~\ref{th:mutinf1}.

\begin{proof}

  \medskip
  \subsubsection*{Concavity}

  The output differential entropy $h_Y(F)$ is a concave function of
  $F$ on $\mathcal{F}$. In fact,
  \begin{equation*}
    h_Y(F)\,=\,-\int\,p_Y(y;F)\,\ln\,p_Y(y;F)\,dy
  \end{equation*}
  exists (by Theorem~\ref{th:mutinf1}) and is a concave function of
  $p_Y(\cdot)$ because $-x\ln x$ is concave in $x$. Since $p_Y(F)$ is
  linear in $F$, $I(F)\,= h_Y(F)\,-\,h_N$ is concave on $\set{P}_A$.
  
  \medskip
  \subsubsection*{Continuity}

  To prove the continuity of $I(F)$, it suffices to show that $h_Y(F)$
  is continuous by virtue of equation~(\ref{mutdef}). To this end, we
  let $F \in \mathcal{P}_A$ and let $\{F_m\}_{m \geq 1}$ be a sequence
  of probability measures in $\mathcal{P}_A$ that converges weakly to
  $F$. 

  In order to apply Theorem~\ref{thconv} and show the convergence of
  $h_Y(F_m)$ to $h_Y(F)$ and hence the weak continuity of $h_Y(F)$ on
  $\mathcal{P}_A$, we establish that the appropriate conditions are
  satisfied:
  \begin{itemize}
  \item[$\bullet$] By definition of weak convergence, since $p_N(y-x)$
    is bounded and continuous (properties A6 and A7), then $p(y;F_m) =
    \int p_N(y-f(x))\,dF_m(x)$ converges point-wise to $p(y;F) =
    \int p_N(y-f(x))\,dF(x)$.
  \item[$\bullet$] The induced output PDF $p(y;F_m)$ is also bounded
    by one.
  \item[$\bullet$] It remains to find a non-negative and
    non-decreasing function, $l: [0,\infty) \to [0,\infty)$
    such that $l(y) = \omega \left( \ln (y) \right)$, and a scalar $L
    >0$ such that equation~(\ref{eq:lcond}) holds for $p(y;F_m)$, $m
    \geq 1$ and $p(y;F)$, a task which we fulfill in what follows.
  \end{itemize}

  For any $y \geq |f(0)|$, let $\set{S} = f^{-1} \left( [|f(0)|, y]
  \right)$ be the inverse image by $f(\cdot)$ of the closed interval
  $[|f(0)|, y]$. Since $f(\cdot)$ is continuous (A4), the set
  $\set{S}$ is closed. It is also bounded because $|f(x)|$ is
  non-decreasing in $|x|$ and tends to infinity (A5). Therefore any
  element in $\set{S}$ is smaller than a positive $t_u$ such that
  $|f(t_u)| = 2y$ and greater than a negative $t_b$ such that
  $|f(t_b)| = 2y$. Such $t_u$ and $t_b$ exist because $f(\cdot)$ is
  continuous. 

  The set is compact and has a maximal value that we denote $z(y) =
  \max \{z: z \in \set{S}\}$. Note that $|f(z(y))| = y$.

  Define the function $\mathcal{C}_{\min}(\cdot)$: $[0,\infty)
  \longrightarrow \Reals$ as follows:
  \begin{equation*}
    \mathcal{C}_{\min}\left(y\right) = \left\{ \begin{array}{ll}
        \displaystyle \min \left\{\mathcal{C}\left( z(y) \right); \mathcal{C}_N(y) \right\} 
        & y \geq |f(0)| \\ 
        \displaystyle 0  &  \text{otherwise},
      \end{array} \right.
  \end{equation*}
  where $\mathcal{C}_N(\cdot)$ is defined in A8. The function
  $\mathcal{C}_{\min}(y)$ is non-negative and non-decreasing on
  $[0,\infty)$ since both $\mathcal{C}(y)$ and $\mathcal{C}_N(\cdot)$
  are non-negative and non-decreasing by properties A2 and A8 and
  $z(y)$ is non-decreasing for $y \geq |f(0)|$. Additionally,
  $\mathcal{C}_{\min}(y) = \omega\left(\ln y \right)$ because
  $\mathcal{C}(x) = \omega\left(\ln\left|f(x)\right|\right)$ (A3) and
  $\mathcal{C}_N (x) = \omega \left( \ln x \right)$ (A8).

  Now, for any $X$ with distribution $F \in \mathcal{P}_{A}$, 
  \begin{align}
    &\quad \,\,\,\Ep{Y}{\mathcal{C}_{\min}\left[\frac{|Y|}{2}\right]}\nonumber\\
    &= \,  \Ep{X,N}{\mathcal{C}_{\min}\left[\frac{|f(X)+N|}{2}\right]} \nonumber\\
    &\leq \, \Ep{X,N}{ \mathcal{C}_{\min} \left[\frac{|f(X)|+|N|}{2}\right]} \label{nondec0}\\
    &= \, \Ep{X,N}{\mathcal{C}_{\min}\left[ \frac{|f(X)| + |N|}{2}\right] \biggl| {\scriptstyle |f(X)| \leq |N|}}
    \P \left( {\scriptstyle |f(X)| \leq |N|} \right)\nonumber\\
    & \,\, + \Ep{X,N}{\mathcal{C}_{\min}\left[\frac{|f(X)|+|N|}{2}\right] \biggl| {\scriptstyle |f(X)| > |N|}}
    \P \left( {\scriptstyle  |f(X)| > |N|} \right) \nonumber\\
    &\leq \,  \Ep{X,N}{\mathcal{C}_{\min}\left(|N|\right) \biggl|  {\scriptstyle |f(X)| \leq |N|}}
    \P \left(  {\scriptstyle |f(X)| \leq |N|} \right)\nonumber\\
    & \quad + \Ep{X,N}{\mathcal{C}_{\min}\left(|f(X)|\right) \biggl|  {\scriptstyle |f(X)| > |N|}} 
    \P \left(  {\scriptstyle |f(X)| > |N|} \right) \label{nondec1}\\
    &\leq \,  \Ep{N}{\mathcal{C}_{\min}\left(|N|\right)} + \Ep{X}{\mathcal{C}_{\min}\left(|f(X)|\right)} \label{pospos}\\
    &\leq \,   \Ep{N}{\mathcal{C}_{N}\left(|N|\right)} + \Ep{X}{\mathcal{C}\left(|X|\right)}  \nonumber\\
    &\leq \,  L_N + A = L < \infty. \label{thok}
  \end{align} 
  where $0 \leq L_N = \Ep{N}{\mathcal{C}_{N}\left(|N|\right)} <
  \infty$ by property A8. Equations~(\ref{nondec0})
  and~(\ref{nondec1}) are justified since
  $\mathcal{C}_{\min}\left(|x|\right)$ is non-decreasing in $|x|$
  and~(\ref{pospos}) is due to the fact that
  $\mathcal{C}_{\min}\left(|x|\right)$ is non-negative. Since the
  value $0 \leq L < \infty$ is independent of the input distribution
  function $F \in \mathcal{P}_A$, then~(\ref{thok}) holds for any
  output variable $Y$, i.e for all $p(y;F)$ where $F \in
  \mathcal{P}_{A}$. Letting $l(y) =
  \mathcal{C}_{\min}\left(\frac{y}{2}\right)$, $y \in [0,\infty)$,
  then equation~(\ref{eq:lcond}) is satisfied for $p(y;F_m)$, $m \geq
  1$ and $p(y;F)$. Therefore, Theorem~\ref{thconv} holds and
  $h_Y(F_m)$ converges to $h_Y(F)$ and hence $h_Y(F)$ is continuous
  which concludes the proof.
\end{proof}

\section{Extensions}
\label{sc:Ext}

The results may be extended to the case where the noise PDF is not
necessarily continuous on $\Reals$. In fact, we weaken condition A6
and we show that Theorem~\ref{th:main} also holds for noise PDFs which
are piece-wise continuous on a countable number of pieces. Note that
under this category fall absolutely continuous noise variables with a
compact support\footnote[8]{we define the support of a random variable
  as being the set of its points of increase, i.e., the set $\{x \in
  \Reals: \Pr(x-\eta < X < x+\eta)>0 \text{ for all } \eta>0 \}$.}
such as the uniform, and also ones that are one-sided such as the
Gamma or the Pareto random variables. We start by noting the
following:
\begin{itemize}
\item It can be seen from the proof of Theorem~\ref{thconv} that
  almost everywhere (a.e.)\footnote[9]{We say that a property holds
    almost everywhere with respect to a measure $\mu$ and we denote it
    $\mu$-a.e. if and only if the measure by $\mu$ of the set where
    the property fails is equal to zero.} point-wise convergence with
  respect to the Lebesgue measure (in addition to C1 and C2) is
  sufficient in order to have convergence of differential entropies.
\item According to the definition of weak
  convergence~\cite[p. 700]{zap07}, one can replace continuous bounded
  test functions by $F$-a.e. continuous functions where $F$ is the
  limit distribution.
\end{itemize} 


We show now that if $p_N(\cdot)$ has a countable number of
discontinuities then weak convergence of the input distributions
implies Lebesgue-a.e. point-wise convergence of the output PDFs:
Denote by $\{a_i\}_{i \geq 1}$ the countable discontinuities of
$p_N(\cdot)$ and by $\{x_i\}_{i \geq 1}$ the discontinuity points of
$F$, which are necessarily countable (see Jordan decomposition lemma
in~\cite[p. 40]{basu12}). Point-wise convergence of the PDFs holds
except at values of $y$ of the form $y_{ij} = a_i -f(x_j)$, $i,j \geq
1$. The fact that the $\{y_{ij}\}$'s are countable proves our
assertion.


\section{Conclusion}
\label{sc:Conc}

Tangible models for communication channels implicitly assume a finite
value for the channel capacity. Knowing that maximizing the
transmission rates is directly related to a constrained maximization
problem, we have derived sufficient conditions for finiteness and
achievability of the capacity of generic memoryless additive noise
channels. The involved conditions on the input-output relationship,
the input cost function and the type of the noise define a wide
collection of models for which finding codes that strive toward
achieving maximum transmission rates is sensible.  The result is
applicable to possibly non-linear channels, to nearly all the widely
known additive noise models and for cost functions that are
``super-logarithmic''. Interestingly, communications at finite rates is
not directly related to an input average-power constraint. Even when
signaling strategies are allowed to have an infinite second moment on
average, transmission rates could not be arbitrarily large. We mention
that while searching for sufficiency, intermediately we derived
conditions under which point-wise convergence of PDFs implies
convergence of differential entropies.




\bibliographystyle{IEEEtran}
\bibliography{papercor}

\begin{thebibliography}{10}
\providecommand{\url}[1]{#1}
\csname url@rmstyle\endcsname
\providecommand{\newblock}{\relax}
\providecommand{\bibinfo}[2]{#2}
\providecommand\BIBentrySTDinterwordspacing{\spaceskip=0pt\relax}
\providecommand\BIBentryALTinterwordstretchfactor{4}
\providecommand\BIBentryALTinterwordspacing{\spaceskip=\fontdimen2\font plus
\BIBentryALTinterwordstretchfactor\fontdimen3\font minus
  \fontdimen4\font\relax}
\providecommand\BIBforeignlanguage[2]{{%
\expandafter\ifx\csname l@#1\endcsname\relax
\typeout{** WARNING: IEEEtran.bst: No hyphenation pattern has been}%
\typeout{** loaded for the language `#1'. Using the pattern for}%
\typeout{** the default language instead.}%
\else
\language=\csname l@#1\endcsname
\fi
#2}}

\bibitem{Fahs12}
J.~Fahs and I.~Abou-Faycal, ``{On the capacity of additive white alpha-stable
  noise channels},'' in \emph{IEEE International Symposium on Information
  Theory}, Cambridge, MA, USA, 2012, pp. 294--298.

\bibitem{Fahs14-1}
{J. Fahs, I. Abou-Faycal}, ``{A Cauchy input achieves the capacity of a Cauchy
  channel under a logarithmic constraint},'' in \emph{IEEE International
  Symposium on Information Theory}, Honolulu, HI, USA, June 29 - July 4 2014.

\bibitem{ElGhannoudi}
H.~El~Ghannudi, L.~Clavier, N.~Azzaoui, F.~Septier, and P.~a. Rolland,
  ``{Stable interference modeling and Cauchy receiver for an IR-UWB ad hoc
  network},'' \emph{Communications, IEEE Transactions on}, vol.~58, no.~6, pp.
  1748 --1757, Jun. 2010.

\bibitem{marcel}
M.~Nassar, K.~Gulati, A.~Sujeeth, N.~Aghasadeghi, B.~Evans, and K.~Tinsley,
  ``Mitigating near-field interference in laptop embedded wireless
  transceivers,'' in \emph{Acoustics, Speech and Signal Processing, 2008.
  ICASSP 2008. IEEE International Conference on}, 2008, pp. 1405 --1408.

\bibitem{hub}
P.~J. Huber, \emph{Robust Statistics}.\hskip 1em plus 0.5em minus 0.4em\relax
  John Wiley \& Sons, 1981.

\bibitem{Shao}
M.~Shao and C.~Nikias, ``Signal processing with fractional lower order moments:
  stable processes and their applications,'' \emph{Proceedings of the IEEE},
  vol.~81, no.~7, pp. 986 --1010, Jul. 1993.

\bibitem{gon}
{J. G. Gonzalez, D. W. Griffith, and G. R. Arce}, ``Zero-order statistics: A
  signal processing framework for very impulsive processes,'' \emph{Proceedings
  of the IEEE Signal Processing Workshop}, pp. 254--258, 1997.

\bibitem{Sha48_1}
C.~E. Shannon, ``A mathematical theory of communication, part i,'' \emph{Bell
  Syst. Tech. J.}, vol.~27, pp. 379--423, 1948.

\bibitem{Sha48_2}
------, ``A mathematical theory of communication, part ii,'' \emph{Bell Syst.
  Tech. J.}, vol.~27, pp. 623--656, 1948.

\bibitem{Gal68}
R.~Gallager, \emph{Information Theory and Reliable Communication}.\hskip 1em
  plus 0.5em minus 0.4em\relax John Wiley \& Sons, Nov. 1968.

\bibitem{SMITH71}
J.~G. Smith, ``{The information capacity of peak and average power constrained
  scalar Gaussian channels},'' \emph{Inform. Contr.}, vol.~18, pp. 203--219,
  1971.

\bibitem{Luenb}
D.~G. Luenberger, \emph{Optimization By Vector Space Methods}.\hskip 1em plus
  0.5em minus 0.4em\relax New York: Wiley, 1969.

\bibitem{Piera2009}
F.~J. Piera and P.~Parada, ``{On convergence properties of Shannon entropy},''
  \emph{Problems of Information Transmission}, vol.~45, no.~2, pp. 75--94,
  2009.

\bibitem{goda2004}
M.~Godavarti and A.~Hero, ``Convergence of differential entropies,'' \emph{IEEE
  Transactions on Information Theory}, vol.~50, no.~1, pp. 1--6, January 2004.

\bibitem{rioul2011}
O.~Rioul, ``Information theoretic proofs of entropy power inequality,''
  \emph{IEEE Transactions on Information Theory}, vol.~57, no.~1, pp. 33--55,
  January 2011.

\bibitem{kapur}
J.~N. Kapur, \emph{Maximum Entropy Models In Science And Engineering}.\hskip
  1em plus 0.5em minus 0.4em\relax Wiley, 1989.

\bibitem{cover}
{T. M. Cover and J. A. Thomas}, \emph{Elements of Information Theory},
  2nd~ed.\hskip 1em plus 0.5em minus 0.4em\relax Wiley, 2006.

\bibitem{fahsj}
{J. Fahs, I. Abou-Faycal}, ``{Using Hermite bases in studying
  capacity-achieving distributions over AWGN channels},'' \emph{Information
  Theory, IEEE Transactions on}, vol.~58, no.~8, August 2012.

\bibitem{ShiryBook}
A.~N. Shiryaev, \emph{Probability}, 2nd~ed.\hskip 1em plus 0.5em minus
  0.4em\relax Springer-Verlag, 1996.

\bibitem{Mun}
J.~R. Munkres, \emph{Topology}, 2nd~ed.\hskip 1em plus 0.5em minus 0.4em\relax
  Prentice Hall, 2000.

\bibitem{klchung}
K.~L. Chung, \emph{A Course in Probability Theory}, 3rd~ed.\hskip 1em plus
  0.5em minus 0.4em\relax San Diego: Academic Press, 2001.

\bibitem{zap07}
{A. M. Zapala}, ``Unbounded mappings and weak convergence of measures,''
  \emph{Statistics and Probability Letters}, vol.~78, pp. 698--706, 2008.

\bibitem{basu12}
{A. K. Basu}, \emph{Measure Theory and Probability}, 2nd~ed.\hskip 1em plus
  0.5em minus 0.4em\relax PHI Learning, 2012.

\end{thebibliography}

\end{document}